\documentclass[runningheads,a4paper,11pt]{llncs}
\usepackage[utf8x]{inputenc}

\usepackage{fullpage}
\usepackage{amsmath,amssymb}
\usepackage{xspace}

\usepackage{todonotes}
\usepackage{hyperref}

\newcommand{\NP}{\ensuremath{\mathsf{NP}}\xspace}

\newcommand{\F}{\ensuremath{\mathcal{F}}\xspace}
\newcommand{\N}{\mathbb{N}}

\newcommand{\Q}{\mathbb{Q}}

\newcommand{\Z}{\mathbb{Z}}

\newcommand{\Oh}{\mathcal{O}}

\newcommand{\norm}[1]{\big\|#1\big\|}
\DeclareMathOperator{\sign}{sign}
\title{Polynomial Kernels for Weighted Problems\thanks{Supported by the Emmy Noether-program of the German Research Foundation (DFG), KR 4286/1, and ERC Starting Grant 306465 (BeyondWorstCase).}}
\titlerunning{Polynomial Kernels for Weighted Problems}
\author{Michael Etscheid
   \and Stefan Kratsch
   \and Matthias Mnich
   \and Heiko R{\"o}glin}
\authorrunning{Etscheid, Kratsch, Mnich, R{\"o}glin}
\institute{Universit{\"a}t Bonn, Institut f{\"u}r Informatik, Bonn, Germany.
  \texttt{\{etscheid@cs.,kratsch@cs.,mmnich@,roeglin@cs.\}uni-bonn.de}}

\begin{document}

\maketitle

\begin{abstract}
  Kernelization is a formalization of efficient preprocessing for \NP-hard problems using the framework of parameterized complexity. Among open problems in kernelization it has been asked many times whether there are deterministic polynomial kernelizations for \textsc{Subset Sum} and \textsc{Knapsack} when parameterized by the number $n$ of items.
  
  We answer both questions affirmatively by using an algorithm for compressing numbers due to Frank and Tardos (Combinatorica 1987). This result had been first used by Marx and V\'egh (ICALP 2013) in the context of kernelization. We further illustrate its applicability by giving polynomial kernels also for weighted versions of several well-studied parameterized problems. Furthermore, when parameterized by the different item sizes we obtain a polynomial kernelization for \textsc{Subset Sum} and an exponential kernelization for \textsc{Knapsack}. Finally, we also obtain kernelization results for polynomial integer programs.
\end{abstract}

\section{Introduction}
\label{sec:introduction}
The question of handling numerical values is of fundamental importance in computer science. Typical issues are precision, numerical stability, and representation of numbers. In the present work we study the effect that the presence of (possibly large) numbers has on weighted versions of well-studied \NP-hard problems. In other words, we are interested in the effect of large numbers on the computational complexity of solving hard combinatorial problems. Concretely, we focus on the effect that weights have on the preprocessing properties of the problems, and study this question using the notion of \emph{kernelization} from parameterized complexity. Very roughly, kernelization studies whether there are problem \emph{parameters} such that any instance of a given \NP-hard problem can be efficiently reduced to an equivalent instance of small size in terms of the parameter. Intuitively, one may think of applying a set of correct simplification rules, but additionally one has a proven size bound for instances to which no rule applies.

The issue of handling large weights in kernelization has been brought up again and again as an important open problem in kernelization \cite{Worker2010OpenProblems,DagstuhlKernelization2012Report,Worker2013OpenProblems,FPTSummerSchool2014OpenProblems}. For example, it is well-known that for the task of finding a vertex cover of at most $k$ vertices for a given unweighted graph $G$ one can efficiently compute an equivalent instance $(G',k')$ such that $G'$ has at most $2k$ vertices. Unfortunately, when the vertices of~$G$ are additionally equipped with positive rational weights and the chosen vertex cover needs to obey some specified maximum weight $W\in\Q$ then it was long unknown how to encode (and shrink) the vertex weights to bitsize polynomial in $k$.
In this direction, Chebl{\'i}k and Chebl{\'i}kov{\'a}~\cite{ChlebikChlebikova2008} showed that an equivalent graph~$G'$ with total vertex weight at most $2w^*$ can be obtained in polynomial time, whereby~$w^*$ denotes the minimum weight of a vertex cover of~$G$.
This, however, does not mean that the size of~$G'$ is bounded, unless one makes the additional assumption that the vertex weights are bounded from below; consequently, their method only yields a kernel with that extra requirement of vertex weights being bounded away from zero.
In contrast, we do not make such an assumption.

Let us attempt to clarify the issue some more.
The task of finding a polynomial kernelization for a weighted problem usually comes down to two parts: (1) Deriving reduction rules that work correctly in the presence of weights. The goal, as for unweighted problems, is to reduce the number of relevant objects, e.g., vertices, edges, sets, etc., to polynomial in the parameter. (2) Shrinking or replacing the weights of remaining objects such that their encoding size becomes (at worst) polynomial in the parameter.
The former part usually benefits from existing literature on kernels of unweighted problems, but regarding the latter only little progress was made.

For a pure weight reduction question let us consider the {\sc Subset Sum} problem.
Therein we are given~$n$ numbers $a_1,\ldots,a_n\in\N$ and a target value $b\in\N$ and we have to determine whether some subset of the~$n$ numbers has sum exactly $b$.
Clearly, reducing such an instance to size polynomial in $n$ hinges on the ability of handling large numbers $a_i$ and~$b$. Let us recall that a straightforward dynamic program solves {\sc Subset Sum} in time $\Oh(nb)$, implying that large weights are to be expected in hard instances.
Harnik and Naor~\cite{HarnikNaor2010} showed that taking all numbers modulo a sufficiently large random prime $p$ of magnitude about $2^{2n}$ produces an equivalent instance with error probability exponentially small in $n$.
(Note that the obtained instance is with respect to arithmetic modulo~$p$.)
The total bitsize then becomes~$\Oh(n^2)$. Unfortunately, this elegant approach fails for more complicated problems than {\sc Subset Sum}.

Consider the {\sc Subset Range Sum} variant of {\sc Subset Sum} where we are given not a single target value $b$ but instead a lower bound $L$ and an upper bound $U$ with the task of finding a subset with sum in the interval $\{L,\ldots,U\}$. 
Observe that taking the values $a_i$ modulo a large random prime faces the problem of specifying the new target value(s), in particular if $U-L>p$ because then every remainder modulo $p$ is possible for the solution.
Nederlof et al.~\cite{NederlofEtAl2012} circumvented this issue by creating not one but in fact a polynomial number of small instances. Intuitively, if a solution has value close to either $L$ or~$U$ then the randomized approach will work well (possibly making a separate instance for target values close to $L$ or $U$).
For solutions sufficiently far from $L$ or $U$ there is no harm in losing a little precision and dividing all numbers by $2$; then the argument iterates.
Overall, because the number of iterations is bounded by the logarithm of the numbers (i.e., their encoding size), this creates a number of instances that is polynomial in the input size, with each instance having size $\Oh(n^2)$; if the initial input is ``yes'' then at least one of the created instances is ``yes''.\footnote{This is usually called a (disjunctive) Turing kernelization.}

To our knowledge, the mentioned results are the only positive results that are aimed directly at the issue of handling large numbers in the context of kernelization.
Apart from these, there are of course results where the chosen parameter bounds the variety of feasible weights and values, but this only applies to integer domains; e.g., it is easy to find a kernel for \textsc{Weighted Vertex Cover} when all weights are positive integers and the parameter is the maximum total weight $k$.
On the negative side, there are a couple of lower bounds that rule out polynomial kernelizations for various weighted and ILP problems, see, e.g., \cite{BodlaenderEtAl2014,Kratsch2013b}.
Note, however, that the lower bounds appear to ``abuse'' large weights in order to build gadgets for lower bound proofs that also include a super-polynomial number of objects as opposed to having just few objects with weights of super-polynomial encoding size. In other words, the known lower bounds pertain rather to the first step, i.e. finding reduction rules that work correctly in the presence of weights, than to the inherent complexity of the numbers themselves.
Accordingly, since 2010 the question for a deterministic polynomial kernelization for \textsc{Subset Sum} or \textsc{Knapsack} with respect to the number of items can be found among open problems in kernelization~\cite{Worker2010OpenProblems,DagstuhlKernelization2012Report,Worker2013OpenProblems,FPTSummerSchool2014OpenProblems}.

Recently, Marx and V\'egh \cite{MarxVegh2013b} gave a polynomial kernelization for a weighted connectivity augmentation problem. As a crucial step, they use a technique of Frank and Tardos~\cite{FrankTardos1987}, originally aimed at obtaining strongly polynomial-time algorithms, to replace rational weights by sufficiently small and equivalent integer weights.
They observe and point out that this might be a useful tool to handle in general the question of getting kernelizations for weighted versions of parameterized problems.
It turns out that, more strongly, Frank and Tardos' result can also be used to settle the mentioned open problems regarding {\sc Knapsack} and {\sc Subset Sum}.
We point out that this is a somewhat circular statement since Frank and Tardos had set out to, amongst others, improve existing algorithms for ILPs, which could be seen as \emph{very general} weighted problems.

\paragraph{Our work.}
We use the theorem of Frank and Tardos~\cite{FrankTardos1987} to formally settle the open problems, i.e., we obtain deterministic kernelizations for {\sc Subset Sum($n$)} and {\sc Knapsack($n$)}, in Sect.~\ref{sec:settlingopenproblemsviathefranktardostheorem}.
Generally, in the spirit of Marx and V\'egh's observation, this allows to get polynomial kernelizations whenever one is able to first reduce the \emph{number of objects}, e.g., vertices or edges, to polynomial in the parameter. The theorem can then be used to sufficiently shrink the weights of all objects such that the \emph{total size} becomes polynomial in the parameter.

Motivated by this, we consider weighted versions of several well-studied parameterized problems, e.g., {\sc $d$-Hitting Set}, {\sc $d$-Set Packing}, and {\sc Max Cut}, and show how to reduce the number of relevant structures to polynomial in the parameter.
An application of Frank and Tardos' result then implies polynomial kernelizations.
We present our small kernels for weighted problems in Sect.~\ref{sec:smallkernelsforweightedparameterizedproblems}.

Next, we consider the {\sc Knapsack} problem and its special case {\sc Subset Sum}, in Sect.~\ref{sec:kernelknapsackproblems}.
For {\sc Subset Sum} instances with only $k$ item sizes, we derive a kernel of size polynomial in~$k$.
This way, we are improving the exponential-size kernel for this problem due to Fellows et al.~\cite{FellowsEtAl2012}.
We also extend the work of Fellows et al. in another direction by showing that the more general {\sc Knapsack} problem is fixed-parameter tractable (i.e. has an exponential kernel) when parameterized by the number $k$ of item sizes, even for unbounded number of item values.
On the other hand, we provide quadratic kernel size lower bounds for general {\sc Subset Sum} instances assuming the Exponential Time Hypothesis \cite{ImpagliazzoEtAl2001}.

Finally, as a possible tool for future kernelization results we show that the weight reduction approach also carries over to polynomial ILPs so long as the maximum degree and the domains of variables are sufficiently small, in Sect.~\ref{sec:integerpolynomialprogrammingwithboundedrange}.

\section{Preliminaries}
\label{sec:preliminaries}

A \emph{parameterized problem} is a language $\Pi\subseteq\Sigma^*\times\N$, where $\Sigma$ is a finite alphabet; the second component~$k$ of instances $(I,k)\in\Sigma^*\times\N$ is called the \emph{parameter}.
A problem $\Pi\subseteq\Sigma^*\times\N$ is \emph{fixed-parameter tractable} if it admits a \emph{fixed-parameter algorithm}, which decides instances $(I,k)$ of~$\Pi$ in time $f(k) \cdot |I|^{\Oh(1)}$ for some computable function~$f$.
The class of fixed-parameter tractable problems is denoted by $\mathsf{FPT}$.
Evidence that a problem~$\Pi$ is unlikely to be fixed-parameter tractable is that~$\Pi$ is $\mathsf{W}[t]$-hard for some $t\in\mathbb N$ or $\mathsf{W}[P]$-hard, where $\mathsf{FPT}\subseteq \mathsf{W}[1]\subseteq \mathsf{W}[2]\subseteq \hdots \subseteq \mathsf{W}[P]$.
To prove hardness of $\Pi$, one can give a \emph{parameterized reduction} from a $\mathsf{W}[\cdot]$-hard problem~$\Pi'$ to~$\Pi$ that maps every instance $I'$ of~$\Pi'$ with parameter $k'$ to an instance $I$ of $\Pi$ with parameter $k \leq g(k')$ for some computable function $g$ such that $I$ can be computed in time $f(k') \cdot |I'|^{\Oh(1)}$ for some computable function~$f$, and $I$ is a ``yes''-instance if and only if $I'$ is.
If~$f$ and $g$ are polynomials, such a reduction is called a \emph{polynomial parameter transformation}.
A problem~$\Pi$ that is $\mathsf{NP}$-complete even if the parameter~$k$ is constant is said to be \emph{para-$\mathsf{NP}$-complete}. 

A \emph{kernelization} for a parameterized problem $\Pi$ is an efficient algorithm that given any instance $(I,k)$ returns an instance $(I',k')$ such that $(I,k)\in\Pi$ if and only if $(I',k')\in\Pi$ and such that $|I'|+k'\leq f(k)$ for some computable function~$f$.
The function $f$ is called the \emph{size} of the kernelization, and we have a polynomial kernelization if $f(k)$ is polynomially bounded in $k$.
It is known that a parameterized problem is fixed-parameter tractable if and only if it is decidable and has a kernelization.
Nevertheless, the kernels implied by this fact are usually of superpolynomial size.
(The size matches the $f(k)$ from the run time, which for \NP-hard problems is usually exponential as typical parameters are upper bounded by the instance size.)
On the other hand, assuming $\mathsf{FPT\neq W[1]}$ no $\mathsf{W[1]}$-hard problem has a kernelization.
Further, there are tools for ruling out polynomial kernels for some parameterized problems~\cite{Drucker2012,BodlaenderEtAl2009} under an appropriate complexity assumption (namely that $\mathsf{NP\nsubseteq coNP/poly}$).
Such lower bounds can be transferred by the mentioned polynomial parameter transformations~\cite{BodlaenderEtAl2011}.

\section{Settling Open Problems via the Frank-Tardos Theorem}
\label{sec:settlingopenproblemsviathefranktardostheorem}

\subsection{Frank and Tardos' theorem}

Frank and Tardos~\cite{FrankTardos1987} describe an algorithm which proves the following theorem.
\begin{theorem}[\cite{FrankTardos1987}]
\label{thm:frank_tardos}
  There is an algorithm that, given a vector $w \in \Q^r$ and an integer~$N$, in polynomial time finds a vector $\overline{w} \in \Z^r$ with~$\norm{\overline{w}}_\infty \leq 2^{4r^3} N^{r(r+2)}$ such that $\sign(w \cdot b) = \sign(\overline{w} \cdot b)$ for all vectors $b\in\mathbb Z^r$ with $\norm{b}_1 \leq N - 1$.
\end{theorem}

This theorem allows us to compress linear inequalities to an encoding length which is polynomial in the number of variables.
Frank and Tardos' algorithm runs even in strongly polynomial time. As a consequence, all kernelizations presented in this work also have a strongly polynomial running time.

\begin{example}
  There is an algorithm that, given a vector $w \in \Q^r$ and a rational $W \in \Q$, in polynomial time finds a vector $\overline{w} \in \Z^r$ with $\norm{\overline{w}}_\infty = 2^{\Oh(r^3)}$ and an integer $\overline{W} \in \Z$ with total encoding length~$\Oh(r^4)$, such that $w \cdot x \leq W $ if and only if $\overline{w} \cdot x \leq \overline{W}$ for every vector $x \in \{0,1\}^r$.
\end{example}
\begin{proof}
  Use Theorem~\ref{thm:frank_tardos} on the vector $(w, W) \in \Q^{r+1}$ with $N = r + 2$ to obtain the resulting vector $(\bar{w}, \bar{W})$. Now let $b = (x, -1) \in \Z^{r+1}$ and note that $\norm{b}_1 \leq N - 1$.
  The inequality $w \cdot x \leq W$ is false if and only if $\sign(w \cdot x - W) = \sign((w, W) \cdot (x, -1)) = \sign((w,W) \cdot b)$ is equal to $+1$. The same holds for $\overline{w} \cdot x \leq \overline{W}$.

  As each $|\bar{w}_i|$ can be encoded with $\Oh(r^3 + r^2 \log N) = \Oh(r^3)$ bits, the whole vector $\bar{w}$ has encoding length $\Oh(r^4)$.
\qed
\end{proof}

\subsection{Polynomial Kernelization for Knapsack}

A first easy application of Theorem~\ref{thm:frank_tardos} is the kernelization of {\sc Knapsack} with the number~$n$ of different items as parameter.

\begin{center}
  \framebox[\textwidth]{
  \begin{tabular}{rl}
    \multicolumn{2}{l}{{\sc Knapsack($n$)}}\\
  \textit{Input:}      & An integer $n\in \N$, rationals $W, P \in \Q$,
                        a weight vector $w \in \Q^n$, and a\\ & profit vector $p \in \Q^n$.\\
  \textit{Parameter:}  & $n$.\\
  \textit{Question:}   & Is there a vector $x \in \{0,1\}^n$ with $w \cdot x \leq W$ and $p \cdot x \geq P$?\\
  \end{tabular}}
\end{center}

\begin{theorem}
  {\sc Knapsack($n$)} admits a kernel of size $\Oh(n^4)$.\qed
\end{theorem}
As a consequence, also {\sc Subset Sum($n$)} admits a kernel of size $\Oh(n^4)$.


\section{Small Kernels for Weighted Parameterized Problems}
\label{sec:smallkernelsforweightedparameterizedproblems}

The result of Frank and Tardos implies that we can easily handle large weights or numbers in kernelization provided that the number of different objects is already sufficiently small (e.g., polynomial in the parameter). In the present section we show how to handle the first step, i.e., the reduction of the number of objects, in the presence of weights for a couple of standard problems. Presumably the reduction in size of numbers is not useful for this first part since the number of different values is still exponential.

\subsection{Hitting Set and Set Packing}

In this section we outline how to obtain polynomial kernelizations for {\sc Weighted $d$-Hitting Set} and {\sc Weighted $d$-Set Packing}.
Since these problems generalize quite a few interesting hitting/covering and packing problems, this extends the list of problems whose weighted versions directly benefit from our results.
The problems are formally defined as follows.

\begin{center}
  \framebox[\textwidth]{
  \begin{tabular}{rl}
    \multicolumn{2}{l}{{\sc Weighted $d$-Hitting Set($k$)}}\\
  \textit{Input:}      & A set family $\F\subseteq \binom{U}{d}$, a function $w\colon U\to \N$, and $k,W\in\N$.\\
  \textit{Parameter:}  & $k$.\\
  \textit{Question:}   & Is there a set $S\subseteq U$ of cardinality at most $k$ and weight
                       $\sum_{u\in S}w(u)\leq W$  such\\ & that $S$ intersects every set in $\F$?
  \end{tabular}\hfill}
\end{center}

\begin{center}
  \framebox[\textwidth]{
  \begin{tabular}{rl}
    \multicolumn{2}{l}{{\sc Weighted $d$-Set Packing($k$)}}\\
  \textit{Input:}      & A set family $\F\subseteq \binom{U}{d}$, a function $w\colon \F\to \N$, and $k,W\in\N$.\\
  \textit{Parameter:}  & $k$.\\
  \textit{Question:}   & Is there a family $\F^*\subseteq\F$ of exactly $k$ disjoint sets of weight
                        $\sum_{F\in \F^*}w(F)\geq W$?
  \end{tabular}\hfill}
\end{center}

Note that we treat~$d$ as a constant. We point out that the definition of {\sc Weighted Set Packing($k$)} restricts attention to exactly $k$ disjoint sets of weight at least $W$. If we were to relax to at least $k$ sets then the problem would be $\NP$-hard already for $k=0$. On the other hand, the kernelization that we present for {\sc Weighted Set Packing($k$)} holds also if we require~$\F^*$ to be of cardinality at most~$k$ (and total weight at least $W$, as before).

Both kernelizations rely on the Sunflower Lemma of Erd\H{o}s and Rado~\cite{ErdosRado1960}, same as their unweighted counterparts.
We recall the lemma.

\begin{lemma}[Erd\H{o}s and Rado~\cite{ErdosRado1960}]
\label{lem:sunflowerlemma}
  Let $\F$ be a family of sets, each of size~$d$, and let $k\in\N$.
  If $|\F|>d!k^d$ then we can find in time $\Oh(|\F|)$ a so-called $k+1$-sunflower, consisting of $k+1$ sets $F_1,\ldots,F_{k+1}\in\F$ such that the pairwise intersection of any two $F_i,F_j$ with $i\neq j$ is the same set~$C$, called the core.
\end{lemma}

For {\sc Weighted $d$-Hitting Set($k$)} we can apply the Sunflower Lemma directly, same as for the unweighted case: Say we are given $(U,\F,w,k,W)$.
If the size of $\F$ exceeds $d!(k+1)^d$ then we find a $(k+2)$-sunflower $\F_s$ in~$\F$ with core~$C$.
Any hitting set of cardinality at most $k$ must contain an element of~$C$.
The same is true for $k+1$-sunflowers so we may safely delete any set $F\in\F_s$ since hitting the set $C\subseteq F$ is enforced by the remaining $k+1$-sunflower. Iterating this reduction rule yields $\F'\subseteq \F$ with $|\F'|=\Oh(k^d)$ and such that $(U,\F,w,k,W)$ and $(U,\F',w,k,W)$ are equivalent.

Now, we can apply Theorem~\ref{thm:frank_tardos}.
We can safely restrict $U$ to the elements $U'$ present in sets of the obtained set family $\F'$, and let $w'=w|_{U'}$.
By Theorem~\ref{thm:frank_tardos} applied to weights $w'$ and target weight $W$ with $N=k+2$ and $r=\Oh(k^d)$ we get replacement weights of magnitude bounded by $2^{\Oh(k^{3d})}N^{\Oh(k^{2d})}$ and bit size~$\Oh(k^{3d})$.
Note that this preserves, in particular, whether the sum of any~$k$ weights is at most the target weight $W$, by preserving the sign of $w_{i_1}+\ldots+w_{i_k}-W$. The total bitsize is dominated by the space for encoding the weight of all elements of the set~$U'$.

\begin{theorem}
  {\sc Weighted $d$-Hitting Set($k$)} admits a kernelization to $\Oh(k^d)$ sets and total size bounded by $\Oh(k^{4d})$.
\end{theorem}

For {\sc Weighted $d$-Set Packing($k$)} a similar argument works.

\begin{theorem}
\label{thm:SetPacking}
  {\sc Weighted $d$-Set Packing($k$)} admits a kernelization to $\Oh(k^d)$ sets and total size bounded by $\Oh(k^{4d})$.
\end{theorem}
\begin{proof}
If the size of~$\F$ exceeds $d!(dk)^d$ then we find a $dk+1$-sunflower $\F_s$ in $\F$ with core $C$. We argue that we can safely discard the set $F_0\in\F_s$ of least weight according to $w\colon\F\to\N$: This could only fail if there is a solution that includes $F_0$, namely $k$ disjoint sets $F_0,\ldots,F_{k-1}$ of total weight at least $W$.
Notice that no set $F_1,\ldots,F_{k-1}$ can contain $C$, since $C\subseteq F_0$. Since $|\F_s|=dk+1$ there must be another set $F_k$, apart from $F_0$, that has an empty intersection with $F_1,\ldots,F_{k-1}$, as the sets in $\F_s$ are disjoint apart from $C$ and there are in total $d(k-1)$ elements in $F_1,\ldots,F_{k-1}$.
It follows that $F_1,\ldots,F_k$ is also a selection of $k$ disjoint sets.
Since~$F_0$ is the lightest set in $\F_s$ we must have that the total weight of $F_1,\ldots,F_k$ is at least $W$.

Iterating this rule gives $|\F|=\Oh(k^d)$.
Again, it suffices to preserve how the sum of any $k$ weights compares with $W$.
Thus, we get the same bound of $\Oh(k^{3d})$ bits per element (of $\F$, in this case).
\qed
\end{proof}

\subsection{Max Cut}

Let us derive a polynomial kernel for {\sc Weighted Max Cut($W$)}, which is defined as follows.

\begin{center}
  \framebox[\textwidth]{
  \begin{tabular}{rl}
    \multicolumn{2}{l}{{\sc Weighted Max Cut($W$)}}\\
  \textit{Input:}      & A graph $G$, a function $w \colon E \to \Q_{\geq 1}$, and $W \in \Q_{\geq 1}$.\\
  \textit{Parameter:}  & $\lceil W\rceil$.\\
  \textit{Question:}   & Is there a set $C \subseteq V(G)$ such that $\sum_{e \in \delta(C)} w(e) \geq W$?
  \end{tabular}\hfill}
\end{center}

Note that we chose the weight of the resulting cut as parameter, which is most natural for this problem. The number~$k$ of edges in a solution is not a meaningful parameter: If we restricted the cut to have at least~$k$ edges, the problem would again be already \NP-hard for~$k=0$. If we required at most~$k$ edges, we could, in this example for integral weights, multiply all edge weights by~$n^2$ and add arbitrary edges with weight~$1$ to our input graph. When setting the new weight bound to~$n^2\cdot W + \binom{n}{2}$, we would not change the instance semantically but there may be no feasible solution left with at most~$k$ edges.

The restriction to edge weights at least $1$ is necessary as otherwise the problem becomes intractable.
This is because when allowing arbitrary positive rational weights, we can transform instances of the $\mathsf{NP}$-complete {\sc Unweighted Max Cut} problem (with all weights equal to~$1$ and parameter $k$, which is the number of edges in the cut) to instances of the {\sc Weighted Max Cut} problem on the same graph with edge weights all equal to $1 / k$ and parameter~$W = 1$.

\begin{theorem}
  {\sc Weighted Max Cut($W$)} admits a kernel of size $\Oh(W^4)$.
\end{theorem}
\begin{proof}
  Let $T$ be the total weight of all edges.
  If $T \geq 2W$, then the greedy algorithm yields a cut of weight at least $T / 2 \geq W$.
  Therefore, all instances with $T \geq 2W$ can be reduced to a constant-size positive instance.    
  Otherwise, there are at most $2W$ edges in the input graph as every edge has weight at least $1$.
  Thus, we can use Theorem~\ref{thm:frank_tardos} to obtain an equivalent (integral) instance of encoding length~$\Oh(W^4)$.
\qed
\end{proof}

\subsection{Polynomial Kernelization for Bin Packing with Additive Error}
\label{sec:polynomialkernelforadditiveonebinpacking}

{\sc Bin Packing} is another classical \NP-hard problem involving numbers.
Therein we are given~$n$ positive integer numbers $a_1,\ldots,a_n$ (the items), a bin size $b\in\N$, and an integer $k$; the question is whether the integer numbers can be partitioned into at most $k$ sets, the bins, each of sum at most~$b$.
From a parameterized perspective the problem is highly intractable for its natural parameter $k$, because for $k=2$ it generalizes the (weakly) \NP-hard {\sc Partition} problem.

Jansen et al.~\cite{JansenEtAl2013} proved that the parameterized complexity improves drastically if instead of insisting on exact solutions the algorithm only has to provide a packing into $k+1$ bins or correctly state that $k$ bins do not suffice.
Concretely, it is shown that this problem variant is fixed-parameter tractable with respect to $k$. The crucial effect of the relaxation is that small items are of almost no importance: If they cannot be added greedily ``on top'' of a feasible packing of big items into~$k+1$ bins, then the instance trivially has no packing into $k$ bins due to exceeding total weight $kb$.
Revisiting this idea, with a slightly different threshold for being a small item, we note that after checking for total weight being at most~$kb$ (else reporting that there is no $k$-packing) we can safely discard all small items before proceeding. Crucially, this cannot turn a no- into a yes-instance because the created $k+1$-packing could then also be lifted to one for all items (contradicting the assumed no-instance).
An application of Theorem~\ref{thm:frank_tardos} then yields a polynomial kernelization because we can have only few large items.

\begin{center}
  \framebox[\textwidth]{
  \begin{tabular}{rl}
    \multicolumn{2}{l}{{\sc Additive One Bin Packing($k$)}}\\
  \textit{Input:}      & Item sizes $a_1,\ldots,a_n\in\N$, a bin size $b\in\N$, and $k\in\N$.\\
  \textit{Parameter:}  & $k$.\\
  \textit{Task:}   & Give a packing into at most $k+1$ bins of size $b$,
    or correctly state that $k$ bins \\ & do not suffice.
  \end{tabular}\hfill}
\end{center}

\begin{theorem}
\label{thm:AdditiveOneBinPacking}
  {\sc Additive One Bin Packing($k$)} admits a polynomial kernelization to $\Oh(k^2)$ items and bit size $\Oh(k^3)$.
\end{theorem}
\begin{proof}
  Let an instance $(a_1,\ldots,a_n,b,k)$ be given.
  If any item size $a_i$ exceeds $b$, or if the total weight of items $a_i$ exceeds $k\cdot b$, then we may safely answer that no packing into $k$ bins is possible.
  In all other cases the kernelization will return an instance whose answer will be correct for the original instance: if it reports a $(k+1)$-packing then the original instance has a $(k+1)$-packing.
  If it reports that no $k$-packing is possible then the same holds for the original instance.

  Assume that the items $a_i$ are sorted decreasingly by value.
  Consider the subsequence, say, $a_1,\ldots,a_\ell$, of items of size at least $\frac{b}{k+1}$.
  If the instance restricted to these items permits a packing into at most~$k+1$ bins, then we show  that the items $a_{\ell+1},\ldots,a_n$ can always be added, giving a $(k+1)$-packing for the input instance: assume that a greedy packing of the small items into the existing packing for $a_1,\ldots,a_\ell$ fails.
  This implies that some item, say $a_i$, of size less than~$\frac{b}{k+1}$ does not fit.
  But then all bins have less than~$\frac{b}{k+1}$ remaining space.
  It follows that the total packed weight, excluding $a_i$, is more than
  \begin{equation*}
    (k+1)\cdot \left(b-\frac{b}{k+1}\right)=(k+1)b-b=kb \enspace .
  \end{equation*}
  This contradicts the fact that this part of the kernelization is only run if the total weight is at most $kb$.
  Thus, a $k+1$-packing for $a_1,\ldots,a_\ell$ implies a $k+1$-packing for the whole set $a_1,\ldots,a_n$.

  Clearly, if the items $a_1,\ldots,a_\ell$ permit no packing into $k$ bins then the same is true for the whole set of items.

  Observe now that $\ell$ cannot be too large: Indeed, since the total weight is at most $kb$ (else we returned ``no'' directly), there can be at most
  \begin{equation*}
    \frac{kb}{\frac{b}{k+1}}=k(k+1)
  \end{equation*}
items of weight at least $\frac{b}{k+1}$.
  Thus, an application of the weight reduction tools yields a total size of~$\Oh(k^3)$.
\qed
\end{proof}


\section{Kernel Bounds for Knapsack Problems}
\label{sec:kernelknapsackproblems}
In this section we provide lower and upper bounds for kernel sizes for variants of the {\sc Knapsack} problem.

\subsection{Exponential Kernel for Knapsack with Few Item Sizes}
First, consider the {\sc Subset Sum} problem restricted to instances with only $k$ distinct item weights, which are not restricted in any other way (except for being non-negative integers).
Then the problem can be solved by a fixed-parameter algorithm for parameter $k$ by a reduction to integer linear programming in fixed dimension, and applying Lenstra's algorithm~\cite{Lenstra1983}
or one of its improvements~\cite{Kannan1987,FrankTardos1987}.
This was first observed by Fellows et al.~\cite{FellowsEtAl2012}.

We now generalize the results by Fellows et al.~\cite{FellowsEtAl2012} to {\sc Knapsack} with few item weights.
More precisely, we are given an instance $I$ of the {\sc Knapsack} problem consisting of $n$ items that have only $k$ distinct item weights; however, the number of item values is unbounded.
This means in particular, that the ``number of numbers'' is not bounded as a function of the parameter, making the results by Fellows et al.~\cite{FellowsEtAl2012} inapplicable.

\begin{theorem}
\label{thm:knapsack-fewitemweights-fpt}
  The {\sc Knapsack} problem with $k$ distinct weights can be solved in time $k^{2{.}5k+o(k)}\cdot \mathrm{poly}(|I|)$, where~$|I|$ denotes the encoding length of the instance.
\end{theorem}
\begin{proof}
Observe that when packing~$x_i$ items of weight $w_i$, it is optimal to pack the~$x_i$ items with largest value among all items of weight $w_i$.
Assume the items of weight~$w_i$ are labeled as $j_1^{(i)},\hdots,j_{n_i}^{(i)}$ by non-increasing values.
For each $s\in\mathbb N$, define $f_i(s) := \sum_{\ell=1}^sv(j_\ell^{(i)})$, where $v(j_\ell^{(i)})$ denotes the value of item~$j_{\ell}^{(i)}$.
We can formulate the knapsack problem as the following program, in which variable~$x_i$ encodes
how many items of weight~$w_i$ are packed into the knapsack and~$g_i$ encodes their total value:
\begin{align*}
               \max~ \sum_{i=1}^kg_i~~~\textnormal{s.t.}~& \sum_{i=1}^kw_i\cdot x_i\leq W,\\
       & g_i \le f_i(x_i),\qquad i=1,\hdots,k,\\
       & x_i \in\{0,1,\ldots,n_i\},~g_i \in\mathbb{N}_0\qquad i=1,\hdots,k.       
\end{align*}

The functions~$f_i$ are in general non-linear. Their concavity implies the following lemma.
\begin{lemma}
  For each~$i$ there exists a set of linear functions~$p_i^{(1)},\ldots,p_i^{(n_i)}$
  such that $f_i(s) = \min_{\ell}p_i^{(\ell)}(s)$ for every~$s \in \{0,\hdots,n_i\}$.
\end{lemma}
\begin{proof}
  For each $\ell \in \{1,\hdots,n_i\}$ we define $p_i^{(\ell)}(s)$ to be the unique linear function such that 
  \begin{equation*}  
    p_i^{(\ell)}(\ell-1) = f_i(\ell-1)\quad\mbox{and}\quad p_i^{(\ell)}(\ell) = f_i(\ell).
  \end{equation*}
  The function~$f_i(s)$ is concave because
  \begin{equation*}
    f_i(\ell+1) - f_i(\ell) = v(j_{\ell+1}^{(i)}) \leq v(j_{\ell}^{(i)}) = f_i(\ell) - f_i(\ell-1)
  \end{equation*}  
  for each $\ell\in\{1,\hdots,n_i-1\}$.
  Therefore, the definition of the linear functions~$p_i^{(\ell)}$ implies that $f_i(s)\leq p_i^{(\ell)}(s)$ for every $\ell\in\{1,\hdots,n_i\}$ and $s\in\{0,\hdots,n_i\}$.
  Since for each $s\in\{1,\hdots,n_i\}$ we have that $p_i^{(s)}(s) = f_i(s)$ and $p_i^{(1)}(0) = f_i(0)$, we conclude that $f_i(s) = \max_{\ell}p_i^{(\ell)}(s)$ for every $s\in\{0,\hdots,n_i\}$.
\qed
\end{proof}

Hence in the program above, we can, for every~$i\in\{1,\hdots,k\}$, replace the 
constraint~$g_i \le f_i(x_i)$ by the set of constraints~$g_i \le p_i^{(\ell)}(x_i)$ 
for~$\ell\in\{1,\ldots,n_i\}$. This way we obtain a formulation of the knapsack problem
as an integer linear program with~$k$ variables. The encoding length of this integer linear program
is polynomially bounded in the encoding length of the instance of {\sc Knapsack}.
Together with the algorithm by Kannan~\cite{Kannan1987}
this implies the fixed-parameter tractability of {\sc Knapsack} with~$k$ item weights.
Using the improved version of this algorithm by Frank and Tardos~\cite{FrankTardos1987},
the theorem follows.\qed
\end{proof}

\subsection{Polynomial Kernel for Subset Sum with Few Item Sizes}
\label{sec:polynomialkernelforsubsetsumwithfewitemweights}
We now improve the work of Fellows et al.~\cite{FellowsEtAl2012} in another direction.
Namely, we show that the {\sc Subset Sum} problem admits a polynomial kernel for parameter the number~$k$ of item sizes; this improves upon the exponential-size kernel due to Fellows et al.~\cite{FellowsEtAl2012}.
To show the kernel bound of~$k^{\Oh(1)}$, consider an instance~$I$ of {\sc Subset Sum} with $n$ items that have only $k$ distinct item sizes.
For each item size $s_i$, let~$\mu_i$ be its multiplicity, that is, the number of items in $I$ of size $s_i$.
Given~$I$, we formulate an ILP for the task of deciding whether some subset $S$ of items has weight exactly~$t$.
The ILP simply models for each item size $s_i$ the number of items $x_i\leq \mu_i$ selected from it as to satisfy the subset sum constraint:
\begin{equation}
\label{eqn:subsetsum_fewsizes_ilp}
  \left.\begin{aligned}
  s_1x_1 + \hdots + s_kx_k = t,&\\
                         0 \leq x_i \leq \mu_i,&\quad i = 1,\hdots,k,\\
                         x_i\in\mathbb N_0,&\quad i = 1,\hdots,k \enspace .
  \end{aligned}\right\rbrace
\end{equation}
Then~\eqref{eqn:subsetsum_fewsizes_ilp} is an {\sc Integer Linear Programming} instance on $m = 1$ relevant constraint and each variable~$x_i$ has maximum range bound $u = \max_i\mu_i\leq n$.

Now consider two cases:
\begin{itemize}
  \item If $\log n\leq k\cdot\log k$, then we apply Theorem~\ref{thm:frank_tardos} to~\eqref{eqn:subsetsum_fewsizes_ilp} to reduce the instance to an equivalent instance~$I'$ of size $\Oh(k^4 + k^3 \log n) = \Oh(k^4 + k^3\cdot (k\log k)) = \Oh(k^4\log k)$.
  We can reformulate~$I'$ as an equivalent {\sc Subset Sum} instance by replacing each size~$s_i$ by $\Oh(\log \mu_i)$ new weights $2^j \cdot s_i$ for $0 \leq j \leq \ell_i$ and $\left(\mu_i - \sum_{j=0}^{\ell_i} 2^j\right) \cdot s_i$, where $\ell_i$ is the largest integer such that $\sum_{j=0}^{\ell_i} 2^j < \mu_i$.
  Then we have $\Oh(k \log n) = \Oh(k^2 \log k)$ items each with a weight which can be encoded in length $\Oh(k^3 + k^2 \log n + \log n) = \Oh(k^3 \log k)$, resulting in an encoding length of $\Oh(k^5 \log^2 k)$.
  \item If $k\log k\leq \log n$, then we solve the integer linear program~\eqref{eqn:subsetsum_fewsizes_ilp}
    by the improved version of Kannan's algorithm~\cite{Kannan1987} due to Frank and Tardos~\cite{FrankTardos1987} that runs in time $d^{2{.}5d+o(d)}\cdot s$ for integer linear programs of dimension $d$ and encoding size $s$.
    As~\eqref{eqn:subsetsum_fewsizes_ilp} has dimension $d = k$ and encoding size $s = |I|$, the condition $k^k\leq n$ means that we can solve the ILP (and hence decide the instance $I$) in time $k^{2{.}5k+o(k)}\cdot s = n^{\Oh(1)}$.
\end{itemize}
In summary, we have shown the following:
\begin{theorem}
\label{thm:subsetsum-fewitemweights-kernel}
  {\sc Subset Sum} with $k$ item sizes admits a kernel of size $\Oh(k^5\log^2 k)$.
  Moreover, it admits a kernel of size $\Oh(k^4 \log k)$ if the multiplicities of the item weights can be encoded in binary.
\end{theorem}
We remark that this method does not work if the instance $I$ is succinctly encoded by specifying the $k$ distinct item weights $w_i$ in binary and for each item size $s_i$ its multiplicity $\mu_i$ in binary: then the running time of Frank and Tardos' algorithm can be exponential in $k$ and the input length of the subset sum instance, which is $\Oh(k \cdot \log n)$.

\subsection{A Kernelization Lower Bound for Subset Sum}
\label{sec:kernelizationlowerboundforsubsetsum}
In the following we show a kernelization lower bound for {\scshape Subset Sum} 
assuming the \emph{Exponential Time Hypothesis}.
The Exponential Time Hypothesis~\cite{ImpagliazzoEtAl2001} states that there does not exist a $2^{o(n)}$-time algorithm for 3-SAT, where~$n$ denotes the number of variables.

\begin{lemma}
\label{lemma:ETHSubsetSum}
  {\sc Subset Sum} does not admit a $2^{o(n)}$-time algorithm assuming the Exponential Time Hypothesis, where~$n$ denotes the number of numbers.
\end{lemma}
\begin{proof}
  The proof is based on a polynomial-time reduction by Gurari~\cite{Gurari1989} that transforms any 3-SAT formula~$\phi$ with~$n$ variables~$v_1,\ldots,v_n$ and~$m$ clauses $C_1,\ldots,C_m$ into an 
equivalent instance of {\scshape Subset Sum} with exactly $2n+2m$ numbers.

  For $j \in \{1,\ldots,m\}$, let clause $C_j = (c_{j1}\vee c_{j2}\vee c_{j3})$, where $c_{j1},c_{j2},c_{j3}\in\{v_1,\neg v_1,\hdots,v_n,\neg v_n\}$. 
  As an intermediate step in the reduction, we consider the following system of linear equations in which we interpret~$v_i$ and~$\neg v_i$ as variables and introduce additional variables~$y_j$ and $y_j'$ for every~$j\in\{1,\ldots,m\}$:
  \begin{equation}
  \label{eqn:knapsackreductions}
    \begin{split}
      \forall i\in\{1,\ldots,n\}: &\quad  v_i + \neg v_i = 1, \\
      \forall j\in\{1,\ldots,m\}: &\quad c_{j1} + c_{j2} + c_{j3} + y_j + y_j' = 3.
    \end{split}
  \end{equation}
  It can easily be checked that this system of linear equations has a solution over~$\{0,1\}$ if and only if the formula~$\phi$ is satisfiable.
  Relabeling the variables yields a reformulation of \eqref{eqn:knapsackreductions} as
  \begin{eqnarray}
  \label{eqn:knapsacksystem}
    \left(\begin{array}{c}
      a_{1,1}\\\vdots\\ a_{n+m,1}
          \end{array}\right)z_1
       + \hdots +
  \left(\begin{array}{c}
    a_{1,2n+2m}\\\vdots\\ a_{n+m,2n+2m}
  \end{array}\right)z_{2n+2m}
  =
    \left(\begin{array}{c}
      c_{1}\\\vdots\\ c_{n+m}
    \end{array}\right),
  \end{eqnarray}
  where~$a_{i,j}\in\{0,1\}$ and~$c_i\in\{1,3\}$.
  We can rewrite this system of equations as the single equation
  \begin{equation}
  \label{eqn:succintknapsack}
    a_1z_1 + \ldots + a_{2n+2m}z_{2n+2m} = C,
  \end{equation}
  where each $a_j\in\N$ is the integer with decimal representation $a_{1,j}\ldots a_{n+m,j}$ and $C\in\N$ denotes the integer with decimal representation~$c_1\ldots c_{n+m}$.
  Equation~\eqref{eqn:succintknapsack} is equivalent to the system~\eqref{eqn:knapsacksystem}, because the sum $a_{i,1} + \hdots + a_{i,2n+2m}$ is at most five.
  This ensures that no carryovers occur and the $h$-th digit of the sum~$a_1z_1 + \ldots + a_{2n+2m}z_{2n+2m}$ is equal to the sum~$a_{h,1}z_1 + \ldots + a_{h,2n+2m}z_{2n+2m}$.
  It follows that \eqref{eqn:knapsackreductions} is satisfiable over $\{0,1\}$ if and only if \eqref{eqn:succintknapsack} is satisfiable over~$\{0,1\}$.

  As a result, the 3-SAT formula $\phi$ is satisfiable if and only if the tuple $(a_1,\hdots,a_{2m+2n},C)$ is 
a ``yes''-instance for {\sc Subset Sum}.
  Now assume there is an algorithm for {\sc Subset Sum} that runs in time~$2^{o(\ell)}$, where~$\ell$ denotes the number of numbers.
  With the reduction above we could use this algorithm to decide whether or not~$\phi$ is satisfiable in time~$2^{o(n+m)}$.
  Due to the sparsification lemma of Impagliazzo et al.~\cite{ImpagliazzoEtAl2001}, this contradicts the Exponential Time Hypothesis.
\qed
\end{proof}

\begin{theorem}
\label{thm:ETHSubsetSum}
  {\sc Subset Sum} does not admit kernels of size~$\Oh(n^{2-\varepsilon})$ for any $\varepsilon > 0$ assuming the Exponential Time Hypothesis, where~$n$ denotes the number of numbers.
\end{theorem}
\begin{proof}
  Assume there exists a kernelization algorithm~$A$ for {\sc Subset Sum} that produces instances of size at most~$\kappa n^{2-\varepsilon}$ for some~$\kappa>0$ and some~$\varepsilon>0$.
  We show that~$A$ can be utilized to solve {\sc Subset Sum} in time~$2^{o(n)}$, which contradicts the Exponential Time Hypothesis due to Lemma~\ref{lemma:ETHSubsetSum}.

  Let~$I$ be an arbitrary {\sc Subset Sum} instance with~$n$ items.
  We apply the kernelization algorithm~$A$ to obtain an equivalent instance~$I'$ whose encoding size is at most~$\kappa n^{2-\varepsilon}$.
  Let~$a_1,\ldots,a_m$ be the numbers in~$I'$ and let~$c$ be the target value.

  Let~$k=n^{1-\varepsilon/2}$. 
  We divide the numbers in~$I'$ into two groups: a number~$a_i$ is called \emph{heavy} if~$a_i\ge 2^{k}$ and \emph{light} otherwise. Since one needs at least~$k$ bits to encode a heavy number, the number of heavy numbers is bounded from above by~$\kappa n^{2-\varepsilon}/k=\kappa n^{1-\varepsilon/2}$.

  We solve instance~$I'$ as follows: for each subset~$J_H$ of heavy numbers, we determine  whether or not there exists a subset~$J_L$ of light numbers such that $\sum_{i\in J_L\cup J_H}a_i=c$ via dynamic programming.
  Since there are at most~$\kappa n^{1-\varepsilon/2}$ heavy numbers, there are at most~$2^{\kappa n^{1-\varepsilon/2}}$ subsets~$J_H$.
  The dynamic programming algorithm runs in time~$\Oh(m^2\cdot 2^{n^{1-\varepsilon/2}})$, as each of the at most~$m$ light numbers is bounded from above by~$2^{n^{1-\varepsilon/2}}$.
  Hence, instance~$I'$ can be solved in time $\Oh(m^2 \cdot 2^{(1+\kappa)n^{1-\varepsilon/2}})=2^{o(n)}$,
where the equation follows because~$m\le \kappa n^{2-\varepsilon}=2^{o(n)}$.
\qed
\end{proof}


\section{Integer Polynomial Programming with Bounded Range}
\label{sec:integerpolynomialprogrammingwithboundedrange}
Up to now, we used Frank and Tardos' result only for linear inequalities with mostly binary variables.
But it also turns out to be useful for more general cases, namely for polynomial inequalities with integral bounded variables.
We use this to show that {\sc Integer Polynomial Programming} instances can be compressed if the variables are bounded.
As a special case, {\sc Integer Linear Programming} admits a polynomial kernel in the number of variables if the variables are bounded.

Let us first transfer the language of Theorem~\ref{thm:frank_tardos} to arbitrary polynomials.
\begin{lemma}
\label{lemma:compress_polynomial}
  Let $f \in \Q[X_1,\ldots,X_n]$ be a polynomial of degree at most $d$ with $r$ non-zero coefficients, and let $u \in \N$.
  Then one can efficiently compute a polynomial $\tilde{f} \in \Z[X_1,\ldots,X_n]$ of encoding length $\Oh\!\left(r^4 + r^3 d \log(r u) + r d \log(nd)\right)$ such that $\mbox{sign}(f(x) - f(y)) = \mbox{sign}(\tilde{f}(x) - \tilde{f}(y))$ for all $x,y \in \{-u,\ldots,u\}^n$.
\end{lemma}
\begin{proof}
  Let $w_1,\ldots,w_r \in \Q$ and $f_1,\ldots,f_r \in \Q[X_1,\ldots,X_n]$ be pairwise distinct monomials with coefficient~$1$ such that $f = \sum_{i=1}^r w_i \cdot f_i$.    
  Apply Theorem~\ref{thm:frank_tardos} to $w = (w_1,\ldots,w_r)$ and $N = 2r u^d + 1$ to obtain $\tilde{w} = (\tilde{w}_1,\ldots,\tilde{w}_r) \in \Z^r$. Set $\tilde{f} = \sum_{i=1}^r \tilde{w}_i \cdot f_i$.

  The encoding length of each $\tilde{w}_i$ is upper bounded by~$\Oh(r^3 + r^2 \log N) = \Oh(r^3 + r^2 \cdot d \cdot \log(r \cdot u))$.
  As there are $\binom{n+d}{d}$ monomials of degree at most $d$, the information to which monomial a coefficient belongs can be encoded in $\Oh(\log((n+d)^d)) = \Oh(d \log(nd))$ bits.
  Hence, the encoding length of $\tilde{f}$ is upper bounded by 
  \begin{equation*}
    \Oh\!\left(r^4 + r^3 d \log(r u) + r d \log(nd)\right).
  \end{equation*}

  To prove the correctness of our construction, let $x,y \in \{-u,\ldots,u\}^n$.
  For $1 \leq i \leq r$, set $b_i = f_i(x) - f_i(y) \in \Z \cap [-2u^d, 2u^d]$, and set $b = (b_1,\ldots,b_r)$.
  Then $\norm{b}_1 \leq r \cdot 2u^d$, and thus by Theorem~\ref{thm:frank_tardos}, $\mbox{sign}(w \cdot b) = \mbox{sign}(\tilde{w} \cdot b)$.
  Then also $\mbox{sign}(f(x) - f(y)) = \mbox{sign}(\tilde{f}(x) - \tilde{f}(y))$, as 
  \begin{align*}
    \tilde{f}(x) - \tilde{f}(y) = \sum_{i=1}^r \tilde{w}_i \cdot (f_i(x) - f_i(y))
                                = \sum_{i=1}^r \tilde{w}_i \cdot b_i
                                = \tilde{w} \cdot b,
  \end{align*}   
  and
  \begin{align*}
    f(x) - f(y) = \sum_{i=1}^r w_i \cdot (f_i(x) - f_i(y))
                = \sum_{i=1}^r w_i \cdot b_i
                = w \cdot b.
  \end{align*}
  This completes the proof of the lemma.
\qed
\end{proof}

We use this lemma to compress {\sc Integer Polynomial Programming} instances.

\begin{center}
  \framebox[\textwidth]{
  \begin{tabular}{rl}
    \multicolumn{2}{l}{{\sc Integer Polynomial Programming}}\\
  \textit{Input:}      & Polynomials $c, g_1, \ldots, g_m \in \Q[X_1,\ldots,X_n]$ of degree at most $d$ encoded by the \\& coefficients of the $\Oh(n^d)$ monomials,
   rationals $b_1,\ldots,b_m, z\in\Q$, and $u \in \N$.\\
 \textit{Question:}   & Is there a vector $x\in \{-u,\ldots,u\}^n$ with $c(x) \leq z$ and
                       $g_i(x) \leq b_i$ for $i = 1,\hdots,m$?\\
  \end{tabular}\hfill}
\end{center}

\begin{theorem}
\label{thm:ippo_compression}
  Every {\sc Integer Polynomial Programming} instance in which~$c$ and each~$g_i$ consist of at most $r$ monomials can be efficiently compressed to an equi\-valent instance with an encoding length that is bounded by $\Oh\!\left(m (r^4 + r^3 d \log(r u) + r d \log(nd))\right)$.
\end{theorem}
\begin{proof}
  Define $c', g_1', \ldots, g_m' \colon \Z^n \times \{0,1\} \to \Q$ as
  \begin{align*}
    c'(x, y) & := c(x) + y \cdot z,\\
  g_i'(x, y) & := g_i'(x) + y \cdot b_i \quad i=1,\hdots,m.
  \end{align*}
  Now apply Lemma~\ref{lemma:compress_polynomial} to $c'$ and $g_1', \ldots, g_m'$ to obtain $\tilde{c'}$ and $\tilde{g}_1', \ldots, \tilde{g}_m'$.
  Thereafter, split these functions up into their parts $(\tilde{c}, \tilde{z})$ and $(\tilde{g}_1, \tilde{b}_1), \ldots, (\tilde{g}_m, \tilde{b}_m)$.
  We claim that the instance $\tilde{I} = (\tilde{c}, \tilde{g}_1, \ldots, \tilde{g}_m, d, \tilde{b}_1, \ldots, \tilde{b}_m, \tilde{z}, u)$ is equivalent to $I$.
  To see this, we have to show that a vector $x \in \{-u,\ldots,u\}^n$ satisfies $c(x) \leq z$ if and only if it satisfies $\tilde{c}(x) \leq \tilde{z}$ (and analogously $g_i(x) \leq b_i$ if and only if $\tilde{g}_i(x) \leq \tilde{b}_i$ for all $i$).
  This follows from
  \begin{align*}
    \mbox{sign}(c(x) - z) & = \mbox{sign}(c'(x,0) - c'(0,1))\\
                          & \overset{(\star)}{=} \mbox{sign}(\tilde{c}'(x,0) - \tilde{c}'(0,1))\\
                          & = \mbox{sign}(\tilde{c}(x) - \tilde{z}),
  \end{align*}
  where equality $(\star)$ follows from Lemma~\ref{lemma:compress_polynomial}.

  It remains to show the upper bound on the encoding length of $I'$.
  Each of the tuples $(c, z),\linebreak (g_1, b_1), \ldots, (g_m, b_m)$ can be encoded with
  \[\Oh\!\left(r^4 + r^3 d \log(r u) + r d \log(nd)\right)\]
  bits.
  The variables $d$ and $u$ can be encoded with $\Oh(\log d + \log u)$ bits.
  In total, this yields the desired upper bound on the kernel size.
\qed
\end{proof}

This way, Theorem~\ref{thm:ippo_compression} extends an earlier result by Granot and Skorin-Karpov~\cite{GranotSkorinKapov1989} who considered the restricted variant of $d = 2$.

As $r$ is bounded from above by $\Oh((n+d)^d)$, Theorem~\ref{thm:ippo_compression} yields a polynomial kernel for the combined parameter $(n,m,u)$ for constant dimensions $d$.
In particular, Theorem~\ref{thm:ippo_compression} provides a polynomial kernel for {\sc Integer Linear Programming} for combined parameter $(n,m,u)$.
This provides a sharp contrast to the result by Kratsch~\cite{Kratsch2013b} that {\sc Integer Linear Programming} does not admit a polynomial kernel for combined parameter $(n,m)$ unless the polynomial hierarchy collapses to the third level.


\section{Conclusion}
\label{sec:conclusion}
In this paper we obtained polynomial kernels for the {\sc Knapsack} problem parameterized by the number of items.
We further provide polynomial kernels for weighted versions of a number of fundamental combinatorial optimization problems, as well as integer polynomial programs with bounded range.
Our small kernels are built on a seminal result by Frank and Tardos about compressing large integer weights to smaller ones.
Therefore, a natural research direction to pursue is to improve the compression quality provided by the Frank-Tardos algorithm.

For the weighted problems we considered here, we obtained polynomial kernels whose sizes are generally larger by some degrees than the best known kernel sizes for the unweighted counterparts of these problems.
It would be interesting to know whether this increase in kernel size as compared to unweighted problems is actually necessary (say it could be that we need more space for objects but also due to space for encoding the weights), or whether the kernel sizes of the unweighted problems can be matched.

\bibliographystyle{plain}
\bibliography{knapsack_bib}

\end{document}